\documentclass[copyright,creativecommons]{eptcs}

\usepackage{amsmath, amssymb, latexsym, amsthm}
\usepackage{multirow}
\usepackage{framed}
\usepackage{color}
\usepackage{tikz}
\usepackage{verbatim}
\usetikzlibrary{patterns}
\usepackage{sidecap}
\usetikzlibrary{decorations.markings}
\usepackage{tikz-cd}
\usetikzlibrary{arrows}
 
\usepackage{booktabs, tabularx}

\usepackage{stmaryrd}

\usepackage{wasysym}
 
\usepackage{hyperref}

\usepackage{eucal}

\usepackage{mathtools}
\mathtoolsset{centercolon}

\newtheorem{innercustomgeneric}{\customgenericname}
\providecommand{\customgenericname}{}
\newcommand{\newcustomtheorem}[2]{%
  \newenvironment{#1}[1]
  {%
   \renewcommand\customgenericname{#2}%
   \renewcommand\theinnercustomgeneric{##1}%
   \innercustomgeneric
  }
  {\endinnercustomgeneric}
}

\newcustomtheorem{customthm}{Theorem}
\newcustomtheorem{customlemma}{Lemma}

\newtheorem{thm}{Theorem}

\newtheorem{theorem}[thm]{Theorem}
\newtheorem{corollary}[thm]{Corollary}
\newtheorem{lemma}[thm]{Lemma}

\newtheorem{example}{Example}
\newtheorem{observation}{Observation}

\newcommand{\EK}{\mathsf{EK}}
\newcommand{\EKB}{\mathsf{EKB}}
\newcommand{\EKK}{\mathsf{EKK}}

\newcommand{\Model}{(X, \cE, I, v)}   

\newcommand{\IModel}{(X, \cE, \oplus, I, v)}   

\newcommand{\cE}{\mathcal{E}}

\newcommand{\cL}{\mathcal{L}}
\newcommand{\cM}{\mathcal{M}}

\newcommand{\cS}{\mathcal{S}}

\newcommand{\cX}{\mathcal{X}}

\newcommand{\aK}{\Box}

\newcommand{\imp}{\rightarrow}

\newcommand{\M}{\hat{K}}
\newcommand{\aM}{\Diamond}
\newcommand{\MB}{\hat{B}}
\newcommand{\ME}{\hat{E}}
\renewcommand{\phi}{\varphi}

\newcommand{\diaast}{\raisebox{-.3ex}{\scalebox{1}{\rotatebox{45}{$\boxast$}}}}

\usepackage{amssymb,tikz}

\newcommand{\draft}[1]{{\color{red}[\textsc{#1}]}}
\newcommand{\defin}[1]{\textbf{#1}}

\newcommand{\lthen}{\rightarrow}
\newcommand{\liff}{\leftrightarrow}
\newcommand{\falsum}{\bot}

\newcommand{\proves}{\vdash}
\newcommand{\defeq}{\coloneqq}

\newcommand{\val}[1]{[\![ #1 ]\!]}

\newcommand{\sval}[1]{\| #1 \|}
\newcommand{\aval}[1]{[\kern-0.25em( #1 )\kern-0.25em]}

\renewcommand{\phi}{\varphi}
\newcommand{\rimp}{\Rightarrow}
\newcommand{\dimp}{\Leftrightarrow}

\newcommand{\commentout}[1]{}

\renewcommand{\L}{\mathcal{L}}

\newcommand{\amods}{\mathrel{\kern.2em|\kern-0.2em{\approx}\kern.2em}}
\newcommand{\notamods}{\mathrel{\kern.2em|\kern-0.2em{\not\approx}\kern.2em}}

\newcommand{\shortv}[1]{}

\title{Uncertainty About Evidence \\  {\small Extended Abstract}}

\author{
Adam Bjorndahl
\institute{Carnegie Mellon University\\
      Pittsburgh, USA}
	\email{abjorn@andrew.cmu.edu}
	\and
Ayb\"{u}ke \"{O}zg\"{u}n
\institute{ILLC, University of Amsterdam \\ Amsterdam, the Netherlands \\
Arch\'{e}, University of St.~Andrews \\ St.~Andrews, Scotland}
	\email{a.ozgun@uva.nl}
}  

\def\titlerunning{Uncertainty About Evidence}
\def\authorrunning{A.~Bjorndahl \& A.~\"Ozg\"un}

\begin{document}

\maketitle

\begin{abstract}
We develop a logical framework for reasoning about knowledge and evidence in which the agent may be uncertain about how to interpret their evidence. 
Rather than representing an evidential state as a fixed subset of the state space, our models allow the set of possible worlds that a piece of evidence corresponds to to vary from one possible world to another, and therefore itself be the subject of uncertainty. 
Such structures can be viewed as (epistemically motivated) generalizations of topological spaces. In this context, there arises a natural distinction between what is \textit{actually} entailed by the evidence and what the agent \textit{knows} is entailed by the evidence---with the latter, in general, being much weaker. We provide a sound and complete axiomatization of the corresponding bi-modal logic of knowledge and evidence entailment, and investigate some natural extensions of this core system, including the addition of a belief modality and its interaction with evidence interpretation and entailment, and the addition of a ``knowability'' modality interpreted via a (generalized) interior operator. 
\end{abstract}

\maketitle

\section{Introduction}


In everyday speech, when we claim, say, that the grass being wet is evidence for its having rained recently, the intended meaning seems to be that seeing the wet grass provides some sort of partial, imperfect, defeasible reason to believe or consider it more likely that it rained recently. In sharp contrast to this, many formal models of information update interpret evidence as being essentially infallible or factive. 
Standard Bayesian updating, for example, tells us that to update a belief (i.e., a probability measure $\pi$) on the basis of an observation (i.e., a subset $E$ of the background state space, the ``evidence''), we should condition $\pi$ on $E$, after which $E$ is assigned probability $1$
(see, e.g., \cite[Chapter 3]{Halpern03}).
For another example,
AGM-style belief revision updates an initial state of knowledge/belief (captured by a set of formulas) on the basis of some new information (i.e., a particular formula $\phi$, the ``evidence'') to produce a new set of formulas that always contains the input formula 
$\phi$ \cite{AGM85,Grove88}. And a variety of logical models for evidence and belief update assume that each piece of evidence corresponds to a set of possible worlds $U$ and entails exactly those propositions $\phi$ such that $U \subseteq \val{\phi}$ \cite{MossP92,DMossP96,MossPS07,vBP11,vBDP12,vBDP14}. Of special note are those models in which the collection of evidence is assumed to take the structure of a \emph{topology} \cite{BOVS17,BO17,Bjorndahl16,BBOS16,Ozgun17}; the framework we propose can be viewed as a generalization of this paradigm.

In this paper we develop a logic for reasoning about evidence that is founded on distinguishing what a given piece of evidence \textit{actually} entails from what it is \textit{believed} (perhaps erroneously) to entail. Thus, in our models, evidence is factive but its \textit{interpretation} can be uncertain. This is accomplished by allowing the set of possible worlds that a given piece of evidence corresponds to to vary across possible worlds, and therefore itself be the subject of uncertainty. Viewed as a generalization of topology, roughly speaking this corresponds to replacing each individual open set with a parametrized family of sets (with the parameter taken from the underlying space itself).

This paves the way for a natural representation of several closely related phenomena including calibration error (where an agent receives a signal from a measurement device but is uncertain or mistaken about how that signal is related to the measured quantity), evidence ``introspection'' failure (where the agent in fact has evidence for $\phi$ but lacks evidence \textit{that} their evidence entails $\phi$), and uncertain margins of error (where the agent has taken a measurement with a certain margin of error, but is unsure what exactly that margin is).

This paper is organized as follows. In the next section we briefly review some of the syntax and notation we will rely on in the rest of the paper. Section \ref{section:evidence-knowledge} introduces \emph{evidence models}, motivating the definitions with intuitions and examples
(Examples \ref{exa:will1} and \ref{exa:will2})
and situating the framework in the context of existing paradigms
(Observations \ref{obs:rel} and \ref{obs:ssm}).
We also provide a sound and complete axiomatization of the corresponding logic of evidence and knowledge.
In Section \ref{section:evidence-belief} we extend evidence models to incorporate belief, discuss some possible relationships between belief and evidence, and axiomatize the resulting logics. Section \ref{section:evidence-knowability} introduces a dynamic component to the models in the form of an evidence combination operation, and explores how this additional structure can be used to define a notion of \emph{knowability}; once again, we provide an axiomatization of the relevant logics. Section \ref{section:conclusion} concludes with a discussion of ongoing work. 
Omitted proofs appear in the full version.


\section{Syntactic Preliminaries}

We first
specify a class of logical languages appropriate for the kinds of reasoning that concern us in this paper.
Given unary modalities $\star_{1}, \ldots, \star_{k}$, let $\L_{\star_{1}, \ldots, \star_{k}}$ denote the propositional language recursively generated by
$$
\phi ::= p \, | \, \lnot \phi \, | \, \phi \wedge \psi \, | \, \star_{i} \phi,
$$
where $p \in \textsc{prop}$, the (countable) set of \emph{primitive propositions}, and $1 \leq i \leq k$. Our focus in this paper is on  the trimodal languages $\L_{K,E,B}$ and $\L_{K,E,\Box}$, and various bi-modal and uni-modal fragments thereof, where we read $K\phi$ as ``the agent knows $\phi$'', $E\varphi$ as ``the evidence entails $\varphi$'', $B\phi$ as ``the agent believes $\phi$'', and finally $\Box \phi$ as ``$\phi$ is knowable'' or ``the agent could come to know $\phi$''. The Boolean connectives $\lor$, $\lthen$, and $\liff$ are defined as usual, and $\falsum$ is defined as an abbreviation for $p \land \lnot p$. We also employ $\M$ as an abbreviation for $\neg K \neg$, $\ME$ for $\neg E\neg$, $\MB$ for $\neg B\neg$, and $\aM$ for $\neg\aK\neg$. 
\begin{table}[htp]
\begin{center}
\begin{tabularx}{\textwidth}{>{\hsize=.6\hsize}X>{\hsize=1.3\hsize}X>{\hsize=1.1\hsize}X}
\toprule
(K$_{\star}$) & $\proves \star(\phi \imp \psi) \imp (\star\phi \imp \star\psi)$ & Distribution\\
(D$_{\star}$) & $\proves \star\phi \imp \lnot\star\lnot\phi$ & Consistency\\
(T$_{\star}$) & $\proves \star\phi \imp \phi$ & Factivity\\
(4$_{\star}$) & $\proves \star\phi \imp \star\star\phi$ & Positive introspection\\
(5$_{\star}$) & $\proves \lnot\star\phi \imp \star\lnot\star\phi$ & Negative introspection\\
(Nec$_{\star}$) & from $\proves \phi$ infer $\proves \star \phi$ & Necessitation\\
\bottomrule
\end{tabularx}
\end{center}
\caption{Some axiom schemes and a rule of inference for $\star$} \label{tbl:axs}
\end{table}%

\vspace{-5mm}

Let $\mathsf{CPL}$ denote an axiomatization of classical propositional logic. Then, following standard naming conventions, we define the following logical systems:
$$
\begin{array}{rcl}
\mathsf{K}_{\star} & = & \mathsf{CPL} \textrm{ + (K$_{\star}$) + (Nec$_{\star}$)}\\
\mathsf{KT}_{\star} & = & \mathsf{K}_{\star} \textrm{ + (T$_{\star}$)}\\
\mathsf{S4}_{\star} & = & \mathsf{KT}_{\star} \textrm{ + (4$_{\star}$)}\\
\mathsf{S5}_{\star} & = &\mathsf{S4}_{\star} \textrm{ + (5$_{\star}$)}\\
\mathsf{KD45}_{\star} & = & \mathsf{K}_{\star} \textrm{ + (D$_{\star}$) + (4$_{\star}$) + (5$_{\star}$)}.
\end{array}
$$

\section{Evidence Models for Evidence Entailment and Knowledge} \label{section:evidence-knowledge}



An \defin{evidence space} is a tuple $(X, \cE, I)$ where $X$ is a nonempty set of \emph{worlds}, $\cE$ is a nonempty set of \emph{evidence states}, and  $I = \{I_{e}\}_{e \in \cE}$ is a parametrized family of functions $I_{e}: X \to 2^{X}$. 
An \defin{evidence model} $\cM=(X, \cE, I, v)$ (over \textsc{prop}) is an evidence space $(X, \cE, I)$ equipped with a \emph{valuation function} $v: \textsc{prop} \to 2^{X}$. 

Intuitively, each $e \in \cE$ represents a ``state of evidence'' the agent may be in---perhaps arising from having made some observation, performed some experiment, found some clue, etc. Crucially, the agent is not conceptualized as being uncertain about \textit{which} state of evidence they find themselves in, but rather about the \textit{interpretation} of any such $e \in \cE$. In particular, the evidence $e$ at world $x$ rules out exactly those worlds outside of $I_{e}(x)$, so $I_{e}(x)$ tells us what the evidence $e$ \textit{actually entails at $x$}.
Call $I_{e}(x)$ an \emph{interpretation of $e$}.
These models therefore differ from many standard representations of evidence: rather than representing evidence directly as subsets of the state space, states of evidence are treated as abstract objects, each of which gets associated, via $I$, to various possible subsets of the state space---that is, various possible interpretations.

Suppose $x$ is the actual world and $e$ is the evidence state the agent is in; then, intuitively,
we should have $x \in I_{e}(x)$, since otherwise the evidence would rule out the actual world, which seems absurd.
Say that $x$ and $e$ are \defin{coherent} when $x \in I_{e}(x)$ and define
$$U_{e} = \{x \in X \: : \: x \in I_{e}(x)\},$$
{\em the collection of worlds that cohere with $e$}.
Since these are precisely the worlds at which the true interpretation of $e$ is compatible with the world, intuitively $U_{e}$ consists of exactly those worlds at which $e$ is a possible state of evidence. Thus, when $x$ and $e$ are coherent (i.e., when $x \in U_{e}$), we call the corresponding pair $(x,e)$ an \defin{evidence scenario}---we think of such pairs as being wholistic, self-consistent descriptions of the world and the agent's state of evidence, analogous to the ``epistemic scenarios'' of subset space logic \cite{MossP92,DMossP96,MossPS07}. 
Similarly to subset space semantics, formulas will be
interpreted in evidence models not at worlds $x$ but at evidence scenarios $(x, e)$. 

We are now in a position to formalize our notion of (actual) evidence entailment. Given an evidence model $\cM=(X, \cE, I, v)$ and an evidence scenario $(x, e)$ in $\cM$, we interpret $\cL_{E}$ in $\cM$ as follows:
$$
\begin{array}{lcl}
(\cM,x,e) \models p & \textrm{ iff } & x \in v(p)\\
(\cM,x,e) \models \lnot \phi & \textrm{ iff } & (\cM,x,e) \not\models \phi\\
(\cM,x,e) \models \phi \land \psi & \textrm{ iff } & (\cM,x,e) \models \phi \textrm{ and } (\cM,x,e) \models \psi\\
(\cM,x,e) \models E \phi & \textrm{ iff } & I_{e}(x) \subseteq \val{\phi}_\cM^{e} \\
\end{array}
$$
where $\val{\phi}_\cM^{e} = \{x \in U_{e} \: : \: (\cM, x,e) \models \phi\}$, the \defin{truth set of $\varphi$ with respect to $e$}. We omit mention of $\cM$ when the model is clear from the context.
A formula $\phi$ is said to be \defin{satisfiable} in an evidence model $\cM$ if there is some evidence scenario $(x, e)$ of $\cM$ such that $(\cM, x, e) \models \phi$ and \defin{valid} in $\cM$ if for all $(x, e)$ of $\cM$, we have $(\cM, x, e) \models \phi$.
Note that, by definition, if $(x,e)$ is an evidence scenario we have $x \in I_{e}(x)$, from which it follows that these semantics validate $E\phi \lthen \phi$; that is, actual evidence entailment is factive.

\emph{Knowledge} is often identified in epistemic models with what follows from the agent's information/evidence. In the present framework, however, this is arguably far too strong, since what \textit{actually} follows from the evidence is a fact about the world that the agent may not have any access to. Somewhat more precisely:
at the world $x$ the evidence $e$ entails $I_{e}(x)$ (as a matter of fact), but of course the agent might be uncertain about which world is the true world, and therefore uncertain about what $e$ actually entails. Nonetheless, even without knowing what the evidence $e$ actually entails, they can at a minimum be certain that \textit{whatever} it entails, the world is compatible with that: in other words, the world is somewhere in $\bigcup_{y \in X}I_{e}(y)$.%
\footnote{``But what if the agent doesn't know that the actual evidence is $e$?'' one might object. But the intent is for \textit{all} uncertainty about the evidence to be encoded in the state space---we take $e$ to be an abstract description of the evidence that is broad enough to be compatible with \textit{every possible interpretation thereof} (as represented by the sets $\{I_{e}(y) \: : \: y \in X\}$).
For example, we might imagine an agent who has pointed their measuring device at a phenomenon they wish to measure, and as a result they now see the number 11.1 on a little display window. Perhaps they do not know the margin of error of this device, or to what degree it is calibrated; perhaps they don't even know what exactly it is measuring! But what they \textit{do} know (plausibly, if we restrict our attention to non-skeptical scenarios), is that the window reads 11.1.}
This motivates the following semantics for knowledge:
$$(\cM, x,e) \models K\phi \; \mbox{iff} \; \bigcup_{y \in X}I_{e}(y) \subseteq \val{\phi}^{e}.$$
In other words, the agent knows $\phi$ just in case $\phi$ follows from \textit{every} interpretation of the evidence. Note that under these semantics, the scheme $K\phi \lthen E\phi$ is valid but its converse is generally not. It is also easy to see that $K$ is an $\mathsf{S5}$ modality, since the set $\bigcup_{y \in X}I_{e}(y)$ does not depend on the state.

Consider now the following natural condition:
\begin{enumerate}
\item[(E1)]
$y \in I_{e}(x) \; \rimp \; y \in I_{e}(y)$.
\end{enumerate}
This simply states that at every world $x$, the evidence entails that it coheres with the world. (E1) implies (in fact is equivalent to) the following:
$$\bigcup_{y \in X}I_{e}(y) = U_{e}.$$
Therefore, under (E1), given evidence $e$, the agent is in a position to know that the world coheres with that evidence. This also implies that in evidence models satisfying (E1), the above semantic clause for knowledge can be equivalently restated as
$$(\cM, x,e) \models K\phi \; \mbox{iff} \; U_{e} \subseteq \val{\phi}^{e}.$$
Note that since $\val{\phi}^{e} \subseteq U_{e}$ by definition, this semantic clause for knowledge is in turn equivalent to $U_{e} = \val{\phi}^{e}$.

Next we observe that evidence models subsume standard relational (Kripke-style) semantics, as well as subset space semantics.
\begin{observation} \label{obs:rel}
\em{
Standard \emph{relational models}
(see, e.g., \cite{BdRV01,FHM95}) of the form $(X, R_E, v)$, where $R_E$ is the accessibility relation for a unary modality $E$, arise as a special case of evidence models when the accessibility relation is reflexive: simply take $\cE = \{e\}$ and define $I_{e}(x) = \{y \in X \: : \: xR_Ey\}$. Then it is easy to see that (E1) is satisfied (since $R_E$ is reflexive) and $x \models E\phi$ (in the relational model, i.e., when $E$ is interpreted by universal quantification over all $R_E$-accessible states) just in case $(x,e) \models E\phi$ (in the evidence model). Moreover, since $U_{e} = X$ in this case, the knowledge modality in the evidence model coincides with the universal modality in the relational model, that is, $(x, e)\models K\phi$ (in the evidence model) iff for all $x\in X$, $x\models \phi$ (in the relational model). This correspondence will play a crucial role in our completeness proof for $\cL_{K, E}$ with respect to evidence models. \qed
}
\end{observation}


\begin{observation} \label{obs:ssm}
\em{
\emph{Subset space models} \cite{MossP92,DMossP96,MossPS07} can also naturally be viewed as special cases of evidence models. Given a subset space model $\cX=(X, \cS, v)$,%
\footnote{That is, $X$ is a nonempty set of states, $\cS$ a collection of subsets of $X$ called \emph{epistemic ranges}, and $v$ a valuation function; formulas are evaluated with respect to \emph{epistemic scenarios} of the form $(x,U)$ where $x \in U \in \cS$, and $(x,U) \models K \phi \dimp U \subseteq \val{\phi}^{U} \defeq \{y \in X \: : \: (y,U) \models \phi\}$.}
we can take $\cE = \{e_{U} \: : \: U \in \cS\}$,
so $\cE$ consists of one evidence state for each $U \in \cS$,
and define $I_{e_{U}}(x) = U$ for all $x\in X$.
Thus, each $e_{U}$ is interpreted uniformly as corresponding to the subset $U$.
Then $\cM=(X, \cE, I, v)$ is an evidence model which clearly satisfies (E1)
(since each $I_{e_{U}}$ is constant)
and we have $U_{e_{U}} = U$. Therefore, $(x,U)$ is an epistemic scenario of $\cX$ if and only if $(x, e_{U})$ is an evidence scenario of $\cM$, and moreover:
$$(\cX, x,U) \models K \phi \; \mbox{iff}  \; (\cM, x,e_{U}) \models K\phi \; \mbox{iff} \; (\cM, x,e_{U}) \models E\phi. \qed$$
}
\end{observation}

\begin{example} \label{exa:will1}
\em{
Consider Williamson's famous clock example
\cite{Williamson00}:
you look at a clock and have a perceptual experience that seems to indicate to you that the minute hand is somewhere on the righthand side of the clock. Let's index the possible positions of the minute hand with the interval $C = [0,60)$ in the obvious way, so for example the state $15$ corresponds to it being quarter-past, $30$ to half-past, etc. The perceptual experience you have is supposed to constitute evidence of some sort, with presumably some margin of error involved. That is, if we call this perceptual experience $e$, we want to say that $e$ doesn't tell us the exact position of the minute hand, but rather guarantees that it must lie in some interval containing the true position in its interior. Call this the \emph{margin of error principle}.

In our framework, we can and will incorporate a further type of uncertainty, namely, uncertainty \textit{about the margin of error}. This seems a very natural type of ignorance to model---after all, we may be sure that our perceptions are not exact without being sure of exactly how inexact they are! One way of capturing this scenario using a simple evidence model $M = (C, \cE, I, v)$ is to define
$$
I_{e}(c) = \left\{ \begin{array}{ll}
\big(\frac{c}{2},\frac{c+30}{2}\big) & \textrm{if $c \in (0,30)$}\\
\emptyset & \textrm{otherwise.}
\end{array} \right.
$$
Then it is easy to see that $c \in I_{e}(c)$ iff $c \in (0,30)$, and moreover for all $c \in C$, $I_{e}(c) \subset (0,30)$, from which (E1) follows. Note also that in every state $c$, $I_{e}(c)$ contains $c$ in its interior, so this model satisfies the margin of error principle.

Suppose our primitive propositions include those in the set $\{pos_{t} \: : \: t \in \{0,1,\ldots,59\}\}$, where $pos_{t}$ is read ``the minute hand is $t$ minutes past twelve'' and $v$ is defined in the obvious way: $v(pos_{t}) = \{t\}$. Clearly, evidence is not ``introspective'' in this model, in the sense that the principle $E\phi \lthen EE\phi$ can fail; for instance, it is easy to see that $(15, e) \models E\lnot pos_{6}$ since $6 \notin (7.5,22.5) = I_{e}(15)$, but $(15,e) \not\models EE\lnot pos_{6}$, since for example $6 \in (4, 19) = I_{e}(8)$ and $8 \in (7.5,22.5) = I_{e}(15)$. That is, at state $15$, the evidence \textit{in fact} rules out that it's 6 minutes past twelve, but doesn't itself guarantee that it rules this out. On the other hand, it is easy to see that $U_{e} = (0,30)$, hence for all $c \in (0,30)$, $(c, e) \models K \lnot pos_{45}$---you are in a position to \textit{know} that the minute hand is not pointing directly to the left, even though you don't know exactly what your evidence entails.\footnote{Incidentally, in this model you are also in a position to know that your evidence is compatible with the hand pointing directly to the right, that is, for all $c \in (0,30)$, $(c,e) \models K \lnot E \lnot pos_{15}$. By contrast, for each $t \neq 15$, there is a $c \in (0,30)$ such that $(c,e) \models E \lnot pos_{t}$.} \qed
}
\end{example}

Of course, in building an epistemic model of this scenario, there is no reason to assume that $C$ itself constitutes the \textit{epistemic} state space. Indeed, doing so leads to the potentially problematic implication that \textit{all} the uncertainty the agent may face is indexed by the position of the clock's minute hand.

\begin{example} \label{exa:will2}
\em{
We consider again the Williamson clock case, except this time we expand the epistemic state space to include not only the possible positions of the minute hand $C$, but also an additional parameter that captures variation in the margin for error. This allows us to ``de-couple'' the margin for error from the actual position of the hand, representing a richer space of epistemic possibilities in which the position of the hand and the margin for error can to some extent vary independently of one another. More precisely, define an evidence model $M' = (C \times (0,1), \cE', I', v')$ where
$$
I'_{e}(c,\mu) = \left\{ \begin{array}{ll}
((1-\mu)c, (1-\mu)c + 30\mu) \times (0,1) & \textrm{if $c \in (0,30)$}\\
\emptyset & \textrm{otherwise.}
\end{array} \right.
$$
Intuitively, $\mu$ captures the (actual) precision of the observation: lower values of $\mu$ correspond to higher precision, and higher values of $\mu$ correspond to lower precision (it is easy to see that the length of the interval $I'_{e}(c,\mu)$ is just $30\mu$). The intervals defined in the previous example arise as the special case where $\mu = \frac{1}{2}$, since $I'_{e}(c,\frac{1}{2}) = \big(\frac{c}{2},\frac{c+30}{2}\big) \times (0,1)$. Notice also that this model assumes that your observation of the clock provides \textit{no} evidence at all pertaining to the precision of that observation as captured by the value of $\mu$, since every possible value of $\mu$ is compatible with every interpretation of $e$. This, of course, is not a required constraint of the present framework, but merely one we find plausible and convenient for the current scenario.

Despite the extra richness in the epistemic state space, this evidence model shares several key properties with the one described in Example \ref{exa:will1}. As before, one can easily check that $(c,\mu) \in I'_{e}(c,\mu)$ iff $c \in (0,30)$, and for all $c \in C$, $I'_{e}(c,\mu) \subset (0,30) \times (0,1)$, so (E1) holds. Moreover, in every state $(c,\mu)$, the first component of $I'_{e}(c,\mu)$ contains $c$ in its interior, so this model also satisfies the margin of error principle. Furthermore, if we interpret $pos_{t}$ in $M'$ in the obvious way (i.e., by setting $v'(pos_{t}) = \{t\} \times (0,1)$), then the scheme $E\phi \lthen EE\phi$ does not hold here either: for example, $((15,.3),e) \models E \lnot pos_{10}$ since $10 \notin (10.5,19.5)$ and $I'_{e}(15,.3) = (10.5,19.5) \times (0,1)$; on the other hand, $((15,.3),e) \not\models EE \lnot pos_{10}$, since $(15,0.4) \in I'_{e}(15,.3)$, and $(10,.3) \in (9,21) \times (0,1) = I'_{e}(15,0.4)$. And finally, analogously to the previous example, we have $U_{e} = (0,30) \times (0,1)$, so for all $c \in (0,30)$ and all $\mu \in (0,1)$, $((c,\mu),e) \models K \lnot pos_{45}$. \qed
}
\end{example}

\subsection{Soundness and Completeness for $\cL_{K, E}$} \label{subsec:ax:EK}

When we interpret $\cL_{E,K}$ in the class of evidence models satisfying (E1), the logic of evidence entailment and knowledge we obtain
is a compound of two familiar logics together with one simple interaction axiom:
$$\EK = \mathsf{S5}_K+\mathsf{KT}_E + \mbox{(KE)},$$
where (KE) denotes the axiom scheme $K\varphi\rightarrow E\varphi$.%
\footnote{This axiom system, as an extension of the normal modal logic $\mathsf{KT}$ with the universal modality, has previously been studied in \cite{GorankoP92} within the standard relational framework. The completeness results obtained therein will help us prove completeness with respect to evidence models.}

\begin{thm}\label{thm:soundness:EK}
$\EK$ is a sound axiomatization of $\L_{E,K}$ with respect to the class of evidence models satisfying (E1).
\end{thm}

Our completeness proof relies on a standard Kripke-style interpretation of $\L_{E,K}$   in relational models and the completeness results pertaining thereto. We therefore begin with a brief review of these notions.

A \defin{relational evidence frame} is a
pair $(X, R_{E})$ where $X$ is a non-empty set and $R_{E}$ is a reflexive, binary relation on $X$.
A \defin{relational evidence model} is a relational evidence frame equipped with a valuation function $v: \textsc{prop} \lthen 2^X$. The language $\L_{E,K}$ is interpreted in a relational  evidence model $M = (X, R_E, v)$ by extending the valuation
function via the standard recursive clauses for the Boolean connectives together with the following:
$$
\begin{array}{lcl}
(M,x) \models E \phi & \textrm{ iff } & R_E(x)\subseteq \sval{\phi}_{M}\\
(M,x) \models K \phi & \textrm{ iff } & X= \sval{\phi}_{M},\\
\end{array}
$$
where  $R_E(x)=\{y\in X : xR_E y\}$ and $\sval{\phi}_{M} = \{x \in X \: : \: (M,x) \models \phi\}$.
Thus $E$ is interpreted by universal quantification over the $R_{E}$-accessible states (as usual for a box-type modality in standard relational semantics), while $K$ is interpreted as a universal modality, as might be expected from Observation \ref{obs:rel}.
We omit mention of $M$ when the model is clear from context.

\begin{thm}[\cite{GorankoP92}]\label{thm:comp:reln1}
$\EK$ is a sound and complete axiomatization of $\L_{E,K}$ with respect to the class of relational evidence models. 
\end{thm}

Given a relational evidence model $M=(X, R_E, v)$, consider the tuple $\cM_M=(X, \{e\},  I, v)$
where $I = \{I_{e}\}$ and for all $x\in X$, $I_e(x)=R_E(x)$.
As shown in Observation \ref{obs:rel}, $\cM_M$ is an evidence model satisfying (E1), and $U_{e} = X$. In particular,
every pair $(x, e)$ is an evidence scenario in $\cM_M$.



\begin{lemma}\label{lem:modal.equiv1}
Let $M=(X, R_E, v)$ be a relational evidence model. Then  for all $\varphi\in \L_{E,K}$  and $x\in X$, we have 
$$M, x\models \varphi \mbox{ iff } \cM_M, (x, e)\models \varphi,$$
where $\cM_M=(X, \{e\},  I, v)$ as described above.
\end{lemma}
\commentout{
\begin{proof} 
The proof follows by induction on the structure of $\varphi$; cases for the primitive propositions and
the Boolean connectives are elementary since their truth value depends only on the actual state, not on the evidence state. So assume inductively that the result holds for $\varphi$; we must
show that it holds also for $E\varphi$ and $K\varphi$. Note that the inductive hypothesis implies that $\|\varphi\|_M = \val{\varphi}^e$.

Case  for $E\varphi$:
\begin{align}
(M, x)\models E\varphi & \mbox{ iff }  R_E(x)\subseteq\sval{\varphi}_M \notag\\
& \mbox{ iff }  I_e(x)\subseteq \val{\varphi}^e \tag{since $R_E(x)=I_e(x)$, IH}\\
& \mbox{ iff } (\cM_M, x, e)\models E\varphi \notag
\end{align}

Case for $K\varphi$:
\begin{align}
(M, x)\models K\varphi & \mbox{ iff }  X = \sval{\varphi}_M \notag\\
& \mbox{ iff }   U_e =  \val{\varphi}^e \tag{since $U_e=X$, IH}  \\
& \mbox{ iff } (\cM_M, x, e)\models K\varphi \notag \qedhere
\end{align}
\end{proof}}

\begin{corollary}\label{cor:comp:evi1}
$\EK$ is a complete axiomatization of $\L_{E,K}$ with respect to the class of  evidence models satisfying (E1).
\end{corollary}
\begin{proof}
This follows from Theorem \ref{thm:comp:reln1} and Lemma \ref{lem:modal.equiv1}: if $\phi \in \L_{E,K}$ is such that $\not\proves_{\EK} \phi$, then by Theorem \ref{thm:comp:reln1}   there is a relational evidence model  $M$ that refutes $\phi$ at some state $x$. Then, by Lemma \ref{lem:modal.equiv1}, $\phi$ is also refuted in $\cM_M$ at the epistemic scenario $(x, e)$,
which completes the proof.
\end{proof}

\section{Evidence Models for Belief} \label{section:evidence-belief}

It is natural to wish to extend the framework we have developed to include a representation not only for knowledge but also \textit{belief}. Defining this extension is relatively straightforward, as it parallels a similar construction from previous work \cite{BO17}. The interest here arises not in the definition itself, but from the subsequent investigation into the interplay between belief and uncertainty about the interpretation of evidence.

A \defin{doxastic evidence model} is simply an evidence model $\cM = (X, \cE, I, v)$ 
in which truth is evaluated with respect to \defin{doxastic evidence scenarios}, which are tuples of the form $(x, e, V)$ where $(x,e)$ is an evidence scenario and 
$\emptyset\not =V \subseteq U_{e}$. The subset $V$ is meant to capture the beliefs of the agent: that is, each $y \in V$ is a world the agent (subjectively) considers possible. Given an evidence model  $\cM = (X, \cE, I, v)$ and a doxastic evidence scenario $(x, e, V)$, the semantic clauses for the primitive propositions and Boolean connectives are as before, while for the modalities we have:
$$
\begin{array}{lcl}
(\cM,x,e, V) \models E \phi & \textrm{ iff } & I_{e}(x) \subseteq \val{\phi}^{e,V} \\
(\cM,x,e, V) \models K \phi & \textrm{ iff } & U_e \subseteq \val{\phi}^{e,V} \\
(\cM,x,e, V) \models B \phi & \textrm{ iff } & V \subseteq \val{\phi}^{e,V} \\
\end{array}
$$
where $\val{\phi}_\cM^{e,V} = \{x \in U_{e} \: : \: (\cM, x,e, V) \models \phi\}$.
Thus, $K$ and $E$ are interpreted essentially as before, while the belief modality $B$ quantifies universally over $V$.
Intuitively, the set $V$
might
be interpreted as the agent's
``conjecture'' about
how the world is,
given the evidence $e$, and
the requirements that $V\not =\emptyset$ and $V \subseteq U_{e}$ guarantee
that the agent does not believe inconsistencies and they believe that their evidence coheres with the world, respectively.\footnote{One can also study agents with possibly inconsistent beliefs in a similar way by simply eliminating the requirement $V\not =\emptyset$.
}
This
corresponds to the very standard ``knowledge-implies-belief'' principle: that is, it makes valid the scheme $K\phi \lthen B\phi$. 
Note also that, just like $U_e$, the doxastic range $V$ is state-independent, which guarantees the validity of the strong introspection principles given in Table \ref{tbl:belief} and Lemma \ref{lem:theorems}.\footnote{More general semantics for $\cL_{E, K, B}$ that do not validate these introspection principles can be obtained by interpreting $B$ with respect to a family of parametrized relations $R = \{R_{e}\}_{e \in \cE}$, where each $R_{e} \subseteq X^{2}$, rather than a fixed $V$ given in a doxastic evidence scenario. Due to the page limit, we leave the details of such a generalization for the extended version of this paper.}

Another constraint one might impose on doxastic evidence scenarios $(x,e,V)$ is the following:
\begin{enumerate}
\item[(E2)]
$y \in V \rimp I_{e}(y) \subseteq V$.
\end{enumerate}
Condition (E2) essentially stipulates that the agent takes evidence entailment seriously: if they consider it possible that the state of evidence leaves open those worlds in $I_{e}(y)$, then they consider each such world possible too. This validates the scheme $B\phi \lthen BE\phi$: if the agent believes $\phi$ then they believe that the evidence entails $\phi$. So it's a kind of ``have responsible beliefs'' constraint: you should only believe that which you believe is entailed by the evidence. 

This condition bears a close resemblance to a principle suggested by Stalnaker \cite{Stalnaker06}, which he called ``strong belief'', namely: $B\phi \lthen BK\phi$, if you belief $\phi$ then you believe that you know it. This essentially makes belief subjectively indistinguishable from knowledge. In the special context of subset space models (Observation \ref{obs:ssm}), or more generally whenever the $K$ and $E$ modalities collapse, our (E2) principle just \textit{is} Stalnaker's ``strong belief'' principle. 
But in general (E2) is weaker: you may believe many things without believing 
that
you \textit{know} them---that is, that they are entailed by \textit{every} interpretation of the evidence---%
instead, what (E2) says is that anything you believe is entailed by \textit{those interpretations of the evidence that you consider possible}.

Even this weaker form of Stalnaker's principle may seem too restrictive, however. Interestingly, it is possible to drop it as a \textit{constraint} on doxastic evidence scenarios without abandoning the intuition entirely. Suppose $(x,e,V)$ is a doxastic evidence scenario; let $V^{1} = V$, and define, for $k > 1$,
$$V^{k} = \bigcup_{y \in V^{k-1}} I_{e}(y)$$
and
$$V^{\infty} = \bigcup_{k=1}^{\infty} V^{k}.$$
Then it is easy to see that $V^{1} \subseteq V^{2} \subseteq \cdots$ is a nested increasing sequence of sets, and $V^{\infty}$ actually does satisfy (E2)---in fact it's the smallest set containing $V$ with this property. We might then interpret $V^{\infty}$ as representing the agent's most ``conservative'' beliefs (so the fact that they satisfy the ``responsibility'' constraint, (E2), makes some sense), whereas $V = V^{1}$ represents the agent's least conservative ``conjecture'', with the sequence $V^{2} \subseteq V^{3} \subseteq \cdots$ bridging the gap between these extremes in a series of discrete jumps or ``levels'' of belief. This is related to the idea of using \emph{plausibility rankings} on possible worlds in order to produce a sequence of beliefs, starting with the ``strongest'' beliefs and gradually weakening them by including less plausible (though still possible) worlds \cite{vD06,HoekMeyer92}.
For example, if we apply this idea to Example \ref{exa:will1} starting with the initial conjecture $V^{1} = (10, 20)$ (corresponding to the belief that the hand of the clock is pointing between the $2$ and the $4$), it is easy to see that $V^{2} = (5,25)$, $V^{3} = (2.5, 27.5)$, $\ldots$, and $V^{\infty} = (0,30)$.
A systematic development of this ``ranked belief'' framework in the context of evidence models is left to the full paper.


\subsection{Soundness and Completeness for $\cL_{E, K, B}$}
In order to distinguish the semantics of $\cL_{E, K, B}$ given with respect to doxastic evidence scenarios from those proposed for $\cL_{E, K}$ in Section \ref{section:evidence-knowledge}, we call the former {\em doxastic-evidence semantics}.  Satisfiability and validity  of a formula in doxastic-evidence semantics is defined the same way as given in Section \ref{section:evidence-knowledge}.

The weakest logic of evidence, knowledge, and belief we consider in this paper, denoted $\EKB$,
is obtained
by strengthening $\EK$ with the additional axiom schemes given in Table \ref{tbl:belief}.
\begin{table}[htp]
\begin{center}
\begin{tabularx}{\textwidth}{>{\hsize=.6\hsize}X>{\hsize=1.3\hsize}X>{\hsize=1.1\hsize}X}
\toprule
(K$_B$) & $B(\varphi\rightarrow \psi)\rightarrow (B\varphi\rightarrow B\psi)$ &  Distribution of belief\\
(D$_B$) & $B\varphi \rightarrow \neg B \neg \varphi $ & Consistency of belief \\
(sPI) & $B\varphi\rightarrow KB\varphi$ & Strong positive introspection\\
(KB) & $K\phi \lthen  B\phi$ & Knowledge implies belief\\
\bottomrule
\end{tabularx}
\end{center}
\caption{Additional axiom schemes for $\EKB$}\label{tbl:belief}
\end{table}%
\begin{lemma}\label{lem:theorems}
Nec$_B$ and $\neg B\varphi \rightarrow K\neg B\varphi$ (strong negative introspection) are derivable in $\mathsf{EKB}$.
\end{lemma}
\commentout{
\begin{proof}
Nec$_B$ follows from Nec$_K$ and (KB), and $\neg B\varphi \rightarrow K\neg B\varphi$ follows $\mathsf{S5}_K$ and (sPI).
\end{proof}}

\begin{thm}\label{thm:soundness:EKB}
$\EKB$ is a sound axiomatization of $\L_{E,K,B}$ with respect to the class of evidence models satisfying (E1) under doxastic-evidence semantics.
\end{thm}
The completeness proof again relies on a standard Kripke-style interpretation of $\L_{E,K,B}$ in relational models and the corresponding relational completeness result. 

A \defin{relational doxastic evidence model} $M=(X, R_E, R_B, v)$  is a relational evidence model $(X, R_E, v)$ equipped with an additional binary relation $R_B$ on $X$ such that  for all $x, y\in X$, $R_B(x)\not =\emptyset$ and $R_B(x)=R_B(y)$. The language $\L_{E,K,B}$ is interpreted in a relational doxastic evidence model $M=(X, R_E, R_B, v)$ as before for $E$ and $K$ (see Section \ref{subsec:ax:EK}); for $B$ we have:
$$
 \begin{array}{lcl}
(M,x) \models B \phi & \textrm{ iff } & R_B(x)\subseteq \sval{\phi}_{M}.\\
\end{array}
$$
As usual we omit mention of the model when it is clear from context.

\begin{thm}\label{thm:comp:reln2}
$\EKB$ is a sound and complete axiomatization of $\L_{E,K,B}$ with respect to the class of relational doxastic evidence models. 
\end{thm}
\begin{proof}
While soundness is a matter of routine validity check, completeness follows from a fairly straightforward canonical model construction where (sPI) guarantees that for all $x, y\in X$, $R_B(x)=R_B(y)$ in the canonical model (see, e.g., \cite[Chapters 4 \& 7]{BdRV01}). For a similar construction for topological subset space semantics, see \cite[pp. 20-21, full paper]{BO17}.
\end{proof}

Given a doxastic relational evidence model $M=(X, R_E, R_B, v)$, we construct the evidence model $\cM_M=(X, \{e\},  I, v)$ satisfying (E1) exactly the same way as in Section \ref{subsec:ax:EK}. Let $V = R_B(x)$ for any $x\in X$ and recall that $U_e=X$. Therefore, as $\emptyset \not =V \subseteq U_e$, every tuple of the form $(x, e, V)$ is a doxastic evidence scenario in $\cM_M$.

\begin{lemma}\label{lem:modal.equiv2}
Let $\cM=(X, R_E, R_B, v)$ be a relational doxastic evidence model. Then,  for all $\varphi\in \L_{E,K,B}$  and $x\in X$, we have 
$$M, x\models \varphi \mbox{ iff } \cM_M, (x, e, V)\models \varphi.$$ 
\end{lemma}
\commentout{
\begin{proof} 
The proof follows by induction on the structure of $\varphi$ similarly to  the proof of Theorem \ref{lem:modal.equiv1}, where cases for the primitive propositions, the Boolean connectives, $E\varphi$, and $K\varphi$ are presented. So assume inductively that the result holds for $\varphi$; we
show that it holds also for $B\varphi$.

Case for $B\varphi$:
\begin{align}
(M, x)\models B\varphi & \mbox{ iff }  R_B(x)\subseteq \sval{\varphi}_\cM \notag\\
& \mbox{ iff }  V\subseteq \val{\varphi}^{e,V} \tag{since $R_B(x)=V$, IH}\\
& \mbox{ iff } (\cM_M, x, e, V)\models B\varphi \notag \qedhere
\end{align}
\end{proof}}

\begin{corollary}\label{cor:comp:evi2}
$\EKB$ is a complete axiomatization of $\L_{E,K,B}$ with respect to the class of evidence models satisfying (E1) under doxastic-evidence semantics.
\end{corollary}
\begin{proof}
Similar to the proof of Corollary \ref{cor:comp:evi1}, by Theorem \ref{thm:comp:reln2} and Lemma \ref{lem:modal.equiv2}. \qedhere

\end{proof}

We also provide an axiomatization of $\cL_{E, K, B}$ for evidence models that satisfy (E2) in addition to (E1).

\begin{thm}\label{thm:comp:reln3}
$\EKB+(B\varphi\lthen BE\varphi)$ is a sound and complete axiomatization of $\L_{E,K,B}$ with respect to the class of evidence models satisfying (E1) and (E2) under doxastic-evidence semantics.
\end{thm}
\commentout{
\begin{proof}
For soundness: let $\cM = (X, \cE, I, v)$ be an evidence model satisfying (E1) and (E2), $(x, e, V)$ an evidence scenario of $\cM$, and $\varphi$ such that $(x, e, V)\models B\varphi$. This means that $V\subseteq \val{\varphi}^{e,V}$. Now let $y\in V$. Since $V\subseteq U_e$,
$(y, e, V)$ is an evidence scenario. By (E2), we have that $I_e(y)\subseteq V$, thus, $I_e(y)\subseteq \val{\varphi}^{e,V}$, i.e., $(y, e, V)\models E\varphi$. As $y$ has been chosen arbitrarily from $V$, we conclude that $(x, e, V)\models BE\varphi$.

\noindent For completeness:  it is easy to verify that $\EKB+(B\varphi\lthen BE\varphi)$ is complete with respect to the class of relational doxastic models satisfying 
\begin{equation}\label{eqn:E3.reln}
\mbox{for all $x, y\in X$, if $xR_By$ then $R_E(y)\subseteq R_B(x)$}\tag{E2$'$}
\end{equation} 
Moreover, given a relation doxastic model $M=(X, R_E, R_B, v)$ satisfying (\ref{eqn:E3.reln}), $\cM_M$ satisfies (E1) and (E2). The result then follows from Theorem \ref{thm:comp:reln3} and Lemma \ref{lem:modal.equiv2} similar to the proof of Corollary \ref{cor:comp:evi2}.
\end{proof}}


\section{Evidence Models for Knowability} \label{section:evidence-knowability}

The logics we have considered so far have been \textit{static} in the sense that they include no mechanism for an agent to update their information in any way. As a first step toward introducing a dynamic component to our setting, we consider a simple mechanic for \textit{changing} the state of evidence. Perhaps the simplest intuition comes from the case of an agent who takes multiple successive measurements---assuming they remember the results of previous measurements, it seems reasonable to represent the final state of evidence as a combination $e_{1} \oplus \cdots \oplus e_{k}$, where $e_{i}$ is the evidence state corresponding to the $i$th observation.

This is captured formally in the definition of an
\defin{evidence interaction model},
which
is a tuple $\cM=\IModel$ where  $(\cE, \oplus)$ is a meet-semilattice, $\Model$ is an evidence model satisfying (E1), and for all $x\in X$ and finite $\cE'\subseteq \cE$, $I_{\oplus\cE'}(x)=\bigcap_{e\in \cE'}I_e(x)$.
A notion of \emph{evidence parthood}, denoted by $\leq$, is given by
\begin{equation}\label{eqn:e.parthood}
\forall e', e\in \cE(e'\leq e \mbox{ iff } e'\oplus e=e').\tag{EP}
\end{equation}
Moreover, it is not difficult to see
that for all finite $\cE'\subseteq \cE$, $U_{\oplus\cE'}=\bigcap_{e\in \cE'}U_e$.

Note the analogy with topological spaces. A topological space has the form $(X,\mathcal{T})$, where $\mathcal{T}$ is a collection of subsets of $X$ called \emph{opens}, often conceived of as the results of possible measurements. Evidence spaces effectively replace the topology $\mathcal{T}$ with the structure $(\cE, I)$, so that in place of open subsets $U$ of $X$ we have \textit{families} of subsets $\{I_{e}(x) \: : \: x \in X\}$ of $X$, one for each $e \in \cE$. Loosely speaking, topological spaces might be viewed as special cases of evidence spaces where each $I_{e}$ is a constant function (cf.~Observation \ref{obs:ssm}).

This analogy is taken a step further with evidence interaction models, since the closure of $\cE$ under the meet operation $\oplus$ parallels the closure of $\mathcal{T}$ under intersection. Thus, it may be easier to think of $(\cE, \oplus, I)$ as the analog of a \textit{basis} for $X$, rather than a full topology. We can also define a kind of \textit{generalized interior operator} in evidence interactions models, and use it to
articulate a notion of \emph{measurability} corresponding to what the agent \textit{could come to know} after taking a sufficiently good measurement or otherwise obtaining a sufficiently strong piece of evidence
(see \cite{Bjorndahl16,BO17,Bjorndahl18}).
Given an evidence interaction model $\IModel$ and an evidence scenario $(x, e)$, we interpret the propositional variables, Boolean connectives, $K$, and $E$ as before, and for $\Box$ we define
$$(\cM, x,e) \models \Box\phi \; \mbox{iff} \;  \exists e'\in \cE (x \in U_{e \oplus e'}\subseteq \val{\varphi}^e).$$
Thus, $\Box \phi$ holds just in case there is some piece of evidence $e'$ that, when combined with the agent's current evidence $e$, would result in knowledge of $\val{\phi}^{e}$.

\subsection{Soundness and Completeness for $\cL_{E, K, \Box}$}

The logic of evidence, knowledge, and knowability is obtained by strengthening $\EK$ as follows
$$\EKK = \EK +\mathsf{S4}_\Box + \mbox{(K$\Box$)},$$
where (K$\Box$) denotes the axiom scheme $K\varphi\rightarrow \Box\varphi$. 



\begin{thm}\label{thm:reln.comp2}
$\EKK$ is a sound axiomatization of $\L_{E,K,\Box}$ with respect to the class of evidence interaction models.
\end{thm}
Similarly to the previous completeness proofs, we prove the completeness of $\EKK$ via a detour to the standard relational interpretation of $\L_{E,K,\Box}$ and its corresponding  relational completeness. More precisely, we rely on the completeness of $\EKK$---under the standard Kripke semantics---with respect to the class of finite models of the form $M=(X, R_E, R_\Box, v)$ where $R_E$ is reflexive and $R_\Box$ is reflexive and transitive. We call such structures \defin{relational evidence and knowability models}. While $K$ and $E$ are interpreted in a relational evidence and knowability model as before, $\Box$ is interpreted, in the standard way, via the accessibility relation 
$R_\Box$:
$$
 \begin{array}{lcl}
(M,x) \models \Box \phi & \textrm{ iff } & R_\Box(x)\subseteq \sval{\phi}_{M}.\\
\end{array}
$$

\begin{thm}\label{thm:comp:reln4}
$\EKK$ is a sound and complete axiomatization of $\L_{E,K,\Box}$ with respect to the class of finite relational evidence and knowability models. 
\end{thm}

We now construct an evidence interaction model $\cM_M=\IModel$ from a finite relational evidence and knowability model $M=(X, R_E, R_\Box, v)$ in such a way that $M$ and $\cM_M$ are point-wise modally equivalent with respect to $\cL_{E, K, \Box}$. While finiteness of the model is not essential, it will simplify our construction. 

Let $M=(X, R_E, R_\Box, v)$ be a finite relational evidence and knowability model and denote the set of all upsets of the preordered set $(X, R_\Box)$ by $Up{(X)}$. Since $X$ is finite, $Up{(X)}$ is finite. Therefore, we can enumerate the elements of $Up{(X)}$ and write $Up{(X)}=\{E_1, \dots, E_n\}$. Note that $X\in Up{(X)}$, so, wlog, we let $X=E_n$. For each element in $Up{(X)}$, we put a corresponding element in $\cE=\{e_1, \dots,  e_n\}$ (so $e_n$ is the evidence state corresponding to $X$---this will become clearer below). We then define an evidence parthood relation $\leq$ on $\cE$ as
$$\forall e_i, e_j\in \cE(e_i\leq e_j \mbox { iff } E_i\subseteq E_j).$$
It is easy to see that $(\cE, \leq)$ is a poset with the top element $e_n$, that is, $e_j\leq e_n$ for all $e_j\in \cE$. We can  define the corresponding  meet, $e_i\oplus e_j$, in a standard way as the greatest lower bound of $\{e_i, e_j\}$ with respect to $\leq$. More generally, for any finite $\cE'\subseteq \cE$, the element $\oplus\cE'$ is the greatest lower bound of $\cE'$ with respect to $\leq$ (see, e.g., \cite{Davey90} for a general introduction to lattice theory).  Finally, for all $x\in X$ and $e_i\in \cE$, set $I_{e_i}(x)=E_i\cap R_E(x)$. Notice that, as $E_n=X$, we have $I_{e_n}(x)=R_E(x)$.

\begin{lemma}\label{lem:interaction model}
Given a finite relational evidence and knowability model $M=(X, R_E, R_\Box, v)$, the structure  $\cM_M=\IModel$ constructed in the above described way is an evidence interaction model.
\end{lemma}
\commentout{
\begin{proof} \
\begin{enumerate}
\item $(\cE, \oplus)$ is a meet-semilattice: by the definitions of $\oplus$.

\item (E1) is satisfied:  let $e_i\in \cE$, and $x, y\in X$ such that $y\in I_{e_i}(x)$. Since $I_{e_i}(x)=E_i\cap R_E(x)$, this in particular implies that $y\in E_i$. Then, since $R_E$ is reflexive, we obtain that $y\in E_i\cap R_E(y)=I_{e_i}(y)$.

\item for all $x\in X$ and finite $\cE'\subseteq \cE$, $I_{\oplus\cE'}(x)=\bigcap_{e\in \cE'}I_e(x)$:  let $x\in X$ and $\cE'=\{e_i, \dots, e_k\}\subseteq \cE$. Since $(\cE, \oplus)$ is a meet-semilattice, we have $\oplus\cE'\in \cE$, thus, $\oplus\cE'=e_j$ for some $j$ with $1\leq j \leq n$. We then obtain
\begin{align}
I_{\oplus\cE'}(x) & = I_{e_j}(x)\notag\\
& =E_j\cap R_E(x) \notag \\
& = (E_i\cap \dots \cap E_k) \cap R_E(x) \tag{Lemma \ref{lem:capcup}}\\
& = (E_i\cap R_E(x))  \cap \dots \cap (E_k \cap R_E(x))\notag\\
& =  I_{e_i}(x) \cap \dots \cap I_{e_k}(x)\tag{by defn. of $I_{e_i}$}\\\
& = \bigcap_{e\in \cE'}I_e(x)\notag
\end{align}
\end{enumerate}
\end{proof}}

\begin{lemma}\label{lem:modal.equiv4}
Let $M=(X, R_E, R_\Box, v)$ be a finite relational evidence and knowability model. Then,  for all $\varphi\in \L_{E,K,\Box}$  and $x\in X$, we have 
$$M, x\models \varphi \mbox{ iff } (\cM_M, x, e_n)\models \varphi.$$
\end{lemma}
\commentout{
\begin{proof} 
The proof follows by induction on the structure of $\varphi$ similarly to  the proof of Theorem \ref{lem:modal.equiv1}, where cases for the primitive propositions, the Boolean connectives, $E\varphi$, and $K\varphi$ are presented. So assume inductively that the result holds for $\varphi$; we
show that it holds also for $\Box\varphi$. Note that the inductive hypothesis implies that $\|\varphi\|_M = \val{\varphi}^{e_n}$.

Case $\Box\varphi$:\\
($\Rightarrow$) Suppose $M, x\models \Box\varphi$, i.e, $R_\Box(x)\subseteq \sval{\varphi}_M$. Since $R_\Box(x)\in Up{(X)}$, thus, $R_\Box(x) =E_i$ for some $i$ with $1\leq i\leq n$, there is $e_i\in \cE$ such that $U_{e_i}=R_\Box(x)$ (by Lemma \ref{lem:U}).  Moreover, since $U_{e_n}=E_n=X$, we obtain $U_{e_i} = U_{e_i}\cap U_{e_n}= U_{e_n\oplus e_i}$. Therefore, by IH, we have $x\in U_{e_n\oplus e_i}\subseteq \val{\varphi}^{e_n}$, i.e., $(\cM_M, x, e_n)\models \Box\varphi$.\\
($\Leftarrow$) Suppose $(\cM_M, x, e_n)\models \Box\varphi$. This means that there is $e_i\in \cE$ such that $x\in U_{e_n\oplus e_i}\subseteq  \val{\varphi}^{e_n}$. Since $(\cE, \oplus)$ is a meet-semilattice, $e_n\oplus e_i\in \cE$. Then, by Lemma \ref{lem:U}, $U_{e_n\oplus e_i}\in Up{(X)}$. Since $R_\Box(x)$ is the smallest upset containing $x$, we have $R_\Box(x) \subseteq U_{e_n\oplus e_i}$. Therefore, by IH, we obtain that $R_\Box(x)\subseteq  \sval{\varphi}_M$, i.e., $M, x\models \Box\varphi$.
\end{proof}}

\begin{corollary}\label{cor:comp:evi4}
$\EKK$ is a complete axiomatization of $\L_{E,K,\Box}$ with respect to the class of evidence interaction models.\end{corollary}
\begin{proof}
Similar to the proof of Corollary \ref{cor:comp:evi1}, by Theorem \ref{thm:comp:reln4} and Lemma \ref{lem:modal.equiv4}.
\end{proof}

\begin{theorem}\label{thm:comp.S4}
$\mathsf{S4}_\Box$ is a sound and complete axiomatization of $\cL_\Box$ with respect to the class of evidence interaction models.
\end{theorem}
\commentout{
\begin{proof}
Soundness is given in Theorem \ref{thm:reln.comp2}. The completeness proof follows similarly to the proof of Corollary \ref{cor:comp:evi4} and uses the fact that $\mathsf{S4}_\Box$ is a sound and complete axiomatization of $\cL_\Box$ with respect to the class of finite reflexive and transitive models of the forms $(X, R_\Box, v)$. Given a finite model of the form $(X, R_\Box, v)$, we construct the corresponding evidence interaction model $\cM_M=\IModel$ as explained earlier, except that  we set for all $x\in X$ and $e_i\in \cE$, $I_{e_i}(x)=E_i$.
\end{proof}}

\section{Further Work} \label{section:conclusion}


We introduced \emph{evidence models} as a means of representing agents who may be uncertain about what their evidence actually entails. 
We also explored
some extensions of this framework that
include belief and knowability.
%
%
There are many interesting avenues to continue this line of work.
%
From a philosophical angle, we believe the framework we have developed here is well-suited to the analysis of a variety of conceptual puzzles that arise when less flexible models of evidence entailment are implicitly relied upon, while on the more mathematical side, it is clear that the logical systems we defined have a variety of natural extensions.

In Section \ref{section:evidence-belief}, for example, we outlined a way of using evidence interpretations to extend an agent's initial conjecture to a graded notion of belief/plausibility. And in Section \ref{section:evidence-knowability}, the account of knowability we provided only scratched the surface of the potential for developing fully dynamic logics atop this foundation.
Consider a public announcement style update mechanic in which knowability plays the role of the precondition of the corresponding announcement, as in \cite{Bjorndahl16}. 
The effect of an announcement is then manifested as a transition from the initial evidence state to a more
informative one, without
requiring
global changes in the given model, as in logics of information dynamics interpreted on  subset space models \cite{WA13b,BO17,vDKO15b,BOVS17,BGOVS18}. The enriched structure owing to the evidence states and their variable interpretations raises the question of whether such a dynamic logic can be reduced to a weaker, static logic, as is often the case in similar settings.  

When we view
evidence models as a generalization of subset space models (recall Observation \ref{obs:ssm}), another natural dynamic extension
suggests itself:
adding the so-called {\em effort modality}, the trademark of subset space logics.  The effort modality, denoted here by
$\diaast \varphi$,
is intended to capture a notion of ``epistemic effort'', such as taking further measurements, and
might be read in the present context as
``$\varphi$ becomes true after some further evidence intake''.
It can then be naturally interpreted on evidence interaction models as
$$(x, e)\models  \diaast \varphi \mbox{ iff } (\exists e'\in \cE)(x\in U_{e\oplus e'} \mbox{ and } (x, e\oplus e')\models \varphi).$$ 

\noindent Incorporating such an operator in the current setting would provide a formal framework in which we could study a truly dynamic notion of knowability via the scheme $\diaast K\varphi$, as opposed to its static counterpart $\Box\varphi$ (see also \cite{Bjorndahl16} for a discussion of the same issue in topological subset space semantics). 
Moreover, the relationship between
$\diaast$
and our static modalities $K, B, E$, and $\Box$ could help further the research on dynamic logics for topological formal learning theory \cite{Kelly96}, initiated by \cite{BGS11,Baltag:2015uo} and further developed within subset space style logics in \cite{BGOVS18}. Such investigations are the subject of ongoing research.

\section*{Acknowledgements}\label{sec:Acknowledgements}
We thank the anonymous reviewers of TARK 2019 for their valuable comments. Ayb\"uke \"Ozg\"un's research was funded by the European Research Council (ERC CoG), Consolidator grant no.~681404, `The Logic of Conceivability'.

\bibliographystyle{eptcs}
\bibliography{Ref}

\end{document}